\documentclass[lettersize,journal]{IEEEtran}
\usepackage{amsmath,amsfonts}
\usepackage{amsthm}

\usepackage{algpseudocode}
\usepackage{algorithm}

\usepackage{graphicx}
\usepackage{hyperref}
\usepackage{multirow}
\usepackage{soul}
\usepackage{tabularx}
\usepackage{textcomp}
\usepackage{xcolor}
\usepackage{xspace}

\usepackage[normalem]{ulem}
\usepackage{enumitem}

\newtheorem{definition}{Definition}
\newtheorem{theorem}{Theorem}

\newtheorem{cor}{Corollary}

\newcommand{\Ra}{\ensuremath{\stackrel{\$}{\leftarrow}{\xspace}}}

\newcommand{\A}{\ensuremath{\mathcal{A}}}
\newcommand{\ZZ}{\ensuremath{\mathbb{Z}}}
\newcommand{\Fq}{\ensuremath{\mathbb{F}_q}}
\newcommand{\EC}{\ensuremath{E(\mathbb{F}_q)}}
\newcommand{\B}{\ensuremath{\mathcal{B}}}
\newcommand{\C}{\ensuremath{\mathcal{C}}}
\usepackage{threeparttable}  
\usepackage{booktabs}
\usepackage{multirow}
\usepackage{bm}
\usepackage{pifont}
\usepackage{subcaption}

\algdef{SE}[DOWHILE]{Do}{doWhile}{\algorithmicdo}[1]{\algorithmicwhile\ #1}

  \usepackage{tikz}
 \usetikzlibrary{plotmarks}
 \usepackage{pgfplots}
 \usetikzlibrary{patterns}
 \pgfplotsset{compat=1.3}
 \usetikzlibrary{backgrounds}
 
\newcommand{\sk}{\ensuremath{\mathsf{sk}}}

   \newcommand{\sgnextract}{{\mathtt{Extract}}}

  \newcommand{\userkg}{{\mathtt{UserKG}}}
  \newcommand{\cmpkey}{{\mathtt{CmpKey}}}

 \newcommand{\IBS}{\ensuremath{\mathtt{IBS}}\xspace}
 \newcommand{\sign}{\ensuremath{\mathtt{Sign}}\xspace}
 \newcommand{\verify}{\ensuremath{\mathtt{Verify}}\xspace}
 \newcommand{\keygen}{\ensuremath{\mathtt{KeyGen}}\xspace}
 \newcommand{\setup}{\ensuremath{\mathtt{Setup}}\xspace}
 \newcommand{\PRF}{\ensuremath{\mathtt{PRF}}\xspace}
 \newcommand{\h}{\ensuremath{\mathtt{H}}\xspace}

\usepackage{etoolbox}
\newcommand{\circled}[2][]{	\tikz[baseline=(char.base)]{		\node[shape = circle, draw, fill=red, color=red, inner sep = .2pt]
		(char) {\phantom{\ifblank{#1}{#2}{#1}}};		\node at (char.center) {\makebox[0pt][c]{\color{white}{#2}}};}}
\robustify{\circled}

\newcommand{\scheme}{\texttt{E2IBS}}

\ifCLASSINFOpdf
            \else
                  \fi

\begin{document}
 
\title{Securing 5G Bootstrapping: A Two-Layer\\ IBS Authentication Protocol}
 
\author{Yilu~Dong,
        Rouzbeh~Behnia,
        Attila~A.~Yavuz,
        and~Syed~Rafiul~Hussain,~\IEEEmembership{Member,~IEEE}}

\maketitle

\begin{abstract}
The lack of authentication during the initial bootstrapping phase between cellular devices and base stations allows attackers to deploy fake base stations and send malicious messages to the devices. These attacks have been a long-existing problem in cellular networks, enabling adversaries to launch denial-of-service (DoS), information leakage, and location-tracking attacks. While some defense mechanisms are introduced in 5G, (e.g., encrypting user identifiers to mitigate IMSI catchers), the initial communication between devices and base stations remains unauthenticated, leaving a critical security gap. To address this, we propose \scheme{}, a novel and efficient two-layer identity-based signature scheme designed for seamless integration with existing cellular protocols. We implement \scheme{} on an open-source 5G stack and conduct a comprehensive performance evaluation against alternative solutions. Compared to the state-of-the-art \texttt{Schnorr-HIBS}, \scheme{} reduces attack surfaces, enables fine-grained lawful interception, and achieves 2x speed in verification, making it a practical solution for securing 5G base station authentication.
\end{abstract}


\IEEEpeerreviewmaketitle

\hfill  
 
\hfill January 30, 2025

\section{Introduction}
\IEEEPARstart{5}{G} cellular networks are revolutionizing connectivity, offering faster speeds, greater bandwidth, and enhanced security compared to previous generations. These advancements are driven by innovative physical layer communication technologies and new security policies. Despite these leaps, 5G retains certain mechanisms from earlier generations to ensure seamless backward compatibility.
One such mechanism is the cell selection procedure used during the initial bootstrapping phase. In this phase, a user device selects a suitable base station that enables it to establish a connection with the core network and subsequently with the Internet.
Base stations periodically broadcast information about the network in \texttt{System Information} messages to announce their presence. Cellular devices scan these broadcast messages and connect to an appropriate base station that meets the cell (re-)selection criteria, which include the received signal strength of the broadcast messages, the base station's acceptability to the device, and the type of services offered by that base station.

Despite being critical to establish the root of trust for the communications between devices and networks, \texttt{System Information} broadcast messages in 5G networks are not authenticated. This allows an adversary to spoof~\cite{dabrowski2014imsi,hussain2019insecure} or tamper with~\cite{yang2019hiding} \texttt{System Information} messages by a fake base station emitting signals with a higher strength than that of the nearby legitimate base stations. After luring the cellular devices to connect to it, the fake base station can then launch security and privacy attacks, including DNS-redirection~\cite{rupprecht2019breaking}, denial-of-service (DoS)~\cite{shaik2015practical, hussain2018lteinspector, hussain20195greasoner, 3GPP:33.809, shaik2018impact, rashid2024state}, downgrade~\cite{shaik2015practical, hussain2018lteinspector, hussain20195greasoner, 3GPP:33.809, shaik2019new} attacks, battery depletion attacks~\cite{shaik2015practical, hussain2018lteinspector, hussain20195greasoner, 3GPP:33.809}, information leak attacks~\cite{hussain2018lteinspector, 5gbasechecker, chlosta20215g, park2022doltest, kim2019touching}, location tracking attacks~\cite{5gbasechecker, chlosta20215g, park2022doltest}, and fingerprinting attacks~\cite{5gbasechecker, park2022doltest, rashid2024state}. Although the 5G specifications \cite{3GPP:21.915} have introduced a new cryptographic scheme for preventing the exposure of cellular device's permanent identifier in plaintext, it does not address the root cause of the fake base station problem, which is the absence of authenticating the \texttt{System Information} broadcast messages. A broadcast authentication scheme is essential for a cellular device to verify the legitimacy of the base station it initially connects to. However, such schemes are currently lacking due to deployment challenges and concerns about backward compatibility. This paper aims to address this gap by proposing a practical authentication mechanism that secures the initial connection bootstrapping process between cellular devices and base stations.

Although symmetric-key primitives (e.g., HMAC) can provide efficient authentication, aside from their inherent key distribution and storage hurdles, they fail to provide public verifiability and non-repudiation. 
A recent study by 3GPP~\cite{3GPP:33.809} and other efforts~\cite{lee2009extended, yi1998optimized, zheng1996authentication, gao2021evaluating} have explored using certificate-based Public Key Infrastructure (PKI) or identity-based signature schemes~\cite{boneh2001short, cheng2017sm9, IEEE1363} to authenticate
base stations. However, these techniques are expensive in terms of communication and computation overheads at both the signer and verifier sides. 
Hussain et al.'s scheme~\cite{hussain2019insecure} demonstrates the viability of PKI-based authentication with optimizations like shorter certificates and an offline-online signature generation mechanism. However, it requires costly cryptographic operations and causes long delays in signature verification at the device. These issues are compounded in 5G networks, where higher frequency radio waves and frequent base station handovers introduce significant communication and computational overheads, making it difficult for existing schemes to be adopted in 5G specifications or deployed by service providers. 

In addition to meeting security and performance requirements, the solution must support lawful interception. Lawful interception enables law enforcement agencies to monitor communications for crime investigation and national security purposes while adhering to strict legal frameworks to balance privacy and accountability. For this, law enforcement agencies must authenticate their base stations by obtaining a key from a legitimate PKG. To uphold legal compliance and transparency, law enforcement should deploy temporary base stations only at authorized locations and times, with authentication keys that expire after use. Without a robust revocation scheme, expired keys could be exploited for unauthorized eavesdropping. This risk is particularly severe in identity-based cryptosystems, where user keys are derived from publicly available information.
Existing base station authentication schemes either neglect this scenario~\cite{hussain2019insecure} or lack efficient key management and revocation mechanisms~\cite{singla2021look}, limiting their practicality in real-world deployments.

In summary, a candidate protocol for authenticating initial broadcast messages in 5G networks must satisfy the following requirements: \textbf{[R1]:} The protocol must be efficient for both the signer and the verifier. It must minimize the computation overhead, especially on the verifier side, to preserve battery life for cellular devices without affecting the quality-of-service (QoS) and strict scheduling constraints of broadcast messages. 
\textbf{[R2]:} It must comply with the restriction on the maximum transmission unit of the broadcast radio messages, which further restricts the maximum size of public keys and signatures. In addition, the authentication protocol should limit the communication overhead due to certificates, signatures, and keys as minimal as possible, since additional bytes in broadcast radio packets transmitted over licensed spectrum induce additional costs to cellular service providers. 
\textbf{[R3]:} The protocol must be resilient against relay attacks by an adversary by just re-transmitting \texttt{System Information} messages from a legitimate base station without changes. \textbf{[R4]:} It can handle roaming scenarios, e.g. when the cellular device is outside the coverage area of its service provider and has to use the network of a partner cellular service provider. \textbf{[R5]:} Finally, the protocol can handle lawful interception for law enforcement agencies. Law enforcement agencies need to deploy their fake base stations to locate criminals and intercept their traffic. If an authentication scheme is deployed, law enforcement agencies must be able to authenticate their base stations to user devices.  

Our prior work proposes a broadcast authentication mechanism, \texttt{Schnorr-HIBS}~\cite{singla2021look}, using a hierarchical identity-based variant of Schnorr signature scheme. The proposed protocol introduces a new entity called core Private Key Generator (PKG) or core-PKG in the authentication server function in the 5G core network. Core-PKG generates public-private key-pairs for the Access and Mobility Management Function (AMF), which is the mobility anchor point in the core network and controls multiple base stations.
The AMF, in turn, generates public-private key-pairs for the base stations under its control. A base station uses its private key to sign \texttt{System Information} broadcast messages by following the proposed signature generation scheme to enable cellular devices to efficiently authenticate broadcast messages.
Although \texttt{Schnorr-HIBS} ensures better security and performance than the state-of-the-art, it has the following limitations: \ding{182} higher overhead in verification and communication 
and more attack surfaces due to the hierarchical architecture design, \ding{183} no practical solution for lawful interception, 
\ding{184} no end-to-end implementation and evaluation of the proposed scheme.

To address these challenges, in this paper, we introduce a novel \emph{verifier-efficient 2-layer Identity-Based Signature} (\scheme{}) scheme based on a highly efficient certificate-based method proposed in \cite{ARIS}. By leveraging the key-additive property of the new scheme, we propose a novel method to avoid the single point of failure common in most identity-based cryptosystems while allowing for fine-grained lawful interception.
Compared to \texttt{Schnorr-HIBS} \cite{singla2021look}, \scheme{} provides the following enhancements:  \ding{182} Simpler design and fewer attack surfaces due to the 2-layer design. \ding{183} Key generation for robust security and fine-grained lawful interception with the new key generation methods as discussed in Algorithm \ref{alg:IBS_lawful}. \ding{184} We incorporated the new scheme on an open-source 5G implementation and evaluated the overhead. \ding{185} 2x faster verification from the improved scheme. 

In summary, this paper makes the following contribution:
\ding{182} A comprehensive characterization of the attacks enabled by fake base stations for cellular networks.
\ding{183} \scheme{}, a verifier-efficient identity-based signature scheme and an authentication protocol, significantly reduces the attack surfaces and minimizes the performance overhead. 
\ding{184} A proof-of-concept implementation of our protocol and the integration with open-source 5G stack. We open-source our implementation \cite{E2IBS-github} with all the alternative schemes used for evaluation to provide a foundation for further research.

\section{Background}
\label{Background}
In this section, we first present some notations that will be used throughout the paper. We then briefly describe the architecture of the 5G cellular networks. We also introduce identity-based signatures (IBS), the basic building block of our authentication protocol.

\noindent \textbf{Notations.}  Given two primes $p$ and $q$, we define a finite field $\Fq$~and a group $\ZZ$.  We use \EC~as an elliptic curve (EC) over $\Fq$, where  $P\in \EC$ is the generator of the points on the curve.  We denote a scalar and a point on a curve with small and capital letters, respectively, e.g.,  $ x\Ra S $ denotes a random uniform selection of  $x$ from a set $S$. $||$ denotes string concatenation.  We define two hash functions  $ \h_1$: $ \{0,1\}^{*}  \rightarrow \{0,1\}^{k|t|}$  and  $ \h_2$: $ \{0,1\}^*  \rightarrow   \{0,1\}^{k'|k|}$, for parameters $t,k,k'\in \mathbb{N}$ defined in Section \ref{sec:scheme}. We view these hash functions as random oracles in our security 
analysis~\cite{RandomOracleModel93}.  We define $\PRF_x(\cdot)\rightarrow \mathbb{Z}_p$ as a Pseudo-Random Function initialized by secret $x$. $|x|$ defines the bit-length of variable $x$, (i.e.,$|x|\gets\log_2(x)$).
 
\subsection{5G Cellular Network Architecture}

A 5G cellular network consists of 3 main components (see Figure~\ref{network_architecture}): User Equipment (UE), Next Generation Radio Access Network (NG-RAN) and the 5G Core Network (5GC). 

\noindent  \textbf{UE} refers to the subscriber device (e.g., cell phone) used to access the cellular network. The UE is provided with a Universal Subscriber Identity Module (USIM) card, provisioned by 
a mobile network operator with the permanent identity of the UE, the Subscription Permanent Identifier (SUPI). 

\noindent \textbf{NG-RAN} consists of base stations that UEs can connect to using radio transmission. The base stations broadcast \texttt{System Information} messages, including a Master Information Block (MIB) and multiple System Information Block (SIB) messages periodically. 
MIB and SIB1 together are referred to as \textit{minimum SI}
to enable further communication between the UE and the base stations. The UE listens for SI messages and connects to the base station with the highest signal strength. 

\begin{figure}[t]
 \centering
        \includegraphics[width=\linewidth]{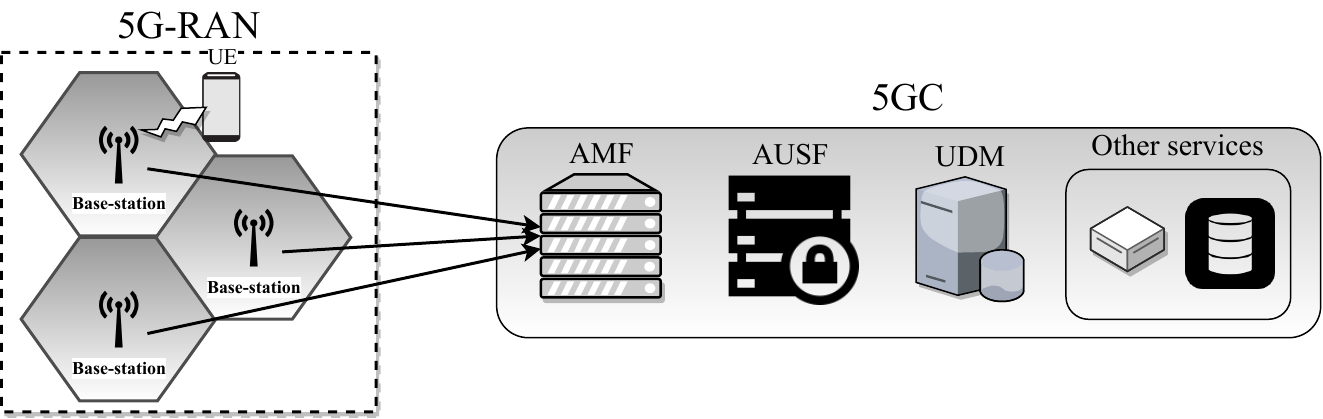}
        \caption{Cellular Network Architecture.}
        \label{network_architecture}
\end{figure}

\noindent  \textbf{5GC} is the brain of the 5G cellular network and houses several components to provide services to the UEs. An important component is the Access and Mobility Management Function (AMF).
supports UE authentication, mobility management, and paging, handles the NAS layer traffic and security, and checks UE's roaming rights. The AMF authenticates the UE in collaboration with the Authentication Server Function (AUSF) and Unified Data Management (UDM). 

\subsection{Identity-Based Signatures}

Digital signatures, achieved via public key infrastructure (PKI), offer several security properties, including message and sender authentication. 
In conventional PKI-based cryptography, the authenticity of public keys is derived through certificates. 
Communicating and verifying these certificates can incur additional overhead; this is especially important for mobile and battery-powered devices.  
Certificate-free cryptography was designed to address this overhead by deriving public keys directly from the user's identifying information, eliminating the need to communicate and verify these certificates. 
This is achieved by enabling a trusted third party (TTP) to compute the user's private key based on their identifying information (e.g., MAC address). 
We formally define the notion of identity-based digital signature in the following definition. 

\begin{definition}\label{def:ibs}

    An identity-based signature scheme $\IBS=(\setup,\keygen,\sign,\verify)$ is defined as follows. 
    \begin{itemize}
        \item  $(msk,params)\gets\IBS.\setup(1^\kappa)$: Given the security parameter $\kappa$, computes TTP's key pair ($msk,mpk$)  and the system parameters $params$. 
                For conciseness, we exclude $params$ as the input of the following algorithms.
                \item $(sk_U,C_U)\gets\IBS.\keygen(msk,U)$: Given the user identity $U$ and $msk$, the TTP computes a commitment value $C_U$ and the secret key $sk_U$.
                \item $\sigma_{m,U}\gets\IBS.\sign(m,sk_U)$: Given a message $m$ and the user's secret key $sk_U$, the user computes the signature $\sigma_{m,U}$.
                \item $\{\text{valid},\text{invalid}\}\gets\IBS.\verify(m,\sigma_{m,U},U,C_U,mpk)$: Given a message-signature pair $(m,\sigma_{m,U})$, the user's identity $U$, its commitment value $C_U$, and the master public key, this algorithm outputs either `valid' or `invalid.' 
            \end{itemize}
   
\end{definition}
The security of an IBS scheme is defined in the following.
\begin{definition}\label{def:eucma}
    The existential unforgeability under the chosen message attack (EU-CMA) for an identity-based signature scheme $\IBS$ is defined via the following experiment $Expt^{\text{EU-CMA}}_{\IBS}$ between a challenger \B~and the adversary \A. 
    \begin{itemize}
        \item The challenger \B~runs $\IBS.\setup()$ and sends $params$, including $mpk$, to \A.
                \item \A~interacts with the $(sk_{U_i},C_{U_i})\gets\mathcal{O}_{Corrupt}({U_i})$ and $(\sigma_{m,{U_i}})\gets\mathcal{O}_{Sign}(m,{U_i})$ oracles for any ${U_i}\in \{U_1, \dots, U_n\}$. 
            \end{itemize}
       \A~succeeds in the above experiment if it outputs a message-signature pair $(m^*,\sigma_{m^*,U^*})$ after a polynomially bounded number of interactions with the above oracles, given $U^* \notin \{U_1, \dots, U_n\}$ and $\sigma_{m^*,U^*}$ was never outputted from $\mathcal{O}_{Sign}(\cdot)$.  
\end{definition}

\begin{definition}\label{def:ecdl} Given an elliptic curve \EC~over a finite field \Fq, and $P,Q\in\EC$, the Elliptic Curve Discrete Logarithm (ECDL) problem asks to find $a\in\ZZ_p$ where $Q=aP \mod q$
    
\end{definition}

\section{\texorpdfstring{Characterization of Attacks Enabled\\ by Fake Base Station}{Characterization of Attacks Enabled by Fake Base Station}}
\label{False Base Stations}
Fake base stations (FBS) have been demonstrated to be feasible in real-world scenarios using Commercial-Off-The-Shelf (COTS) hardware and open-source cellular software stacks~\cite{strobel2007imsi, paget2010practical, kune2012location, shaik2015practical}. To operate a fake base station, an attacker configures their radio to transmit signals at a higher strength than legitimate base stations, enticing users to connect to it instead. Figure~\ref{FBS_a} illustrates a typical fake base station setup used for conducting off-path attacks, while Figure~\ref{FBS_b} depicts its configuration for executing Man-in-the-Middle (MitM) relay attacks.

\begin{figure}[t]
 \centering
    \begin{subfigure}{1\linewidth}
    \includegraphics[width=\linewidth]{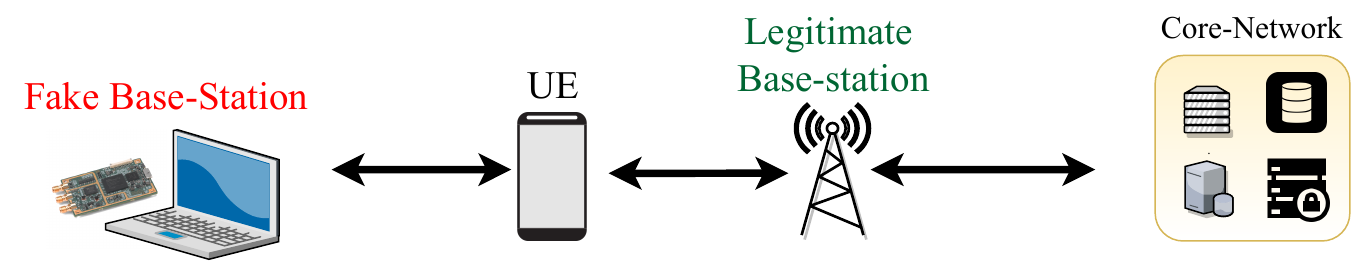}
        \caption{Setup for off-path attacks.}
        \label{FBS_a}
    \end{subfigure}
    
    \begin{subfigure}{1\linewidth}
    \includegraphics[width=\linewidth]{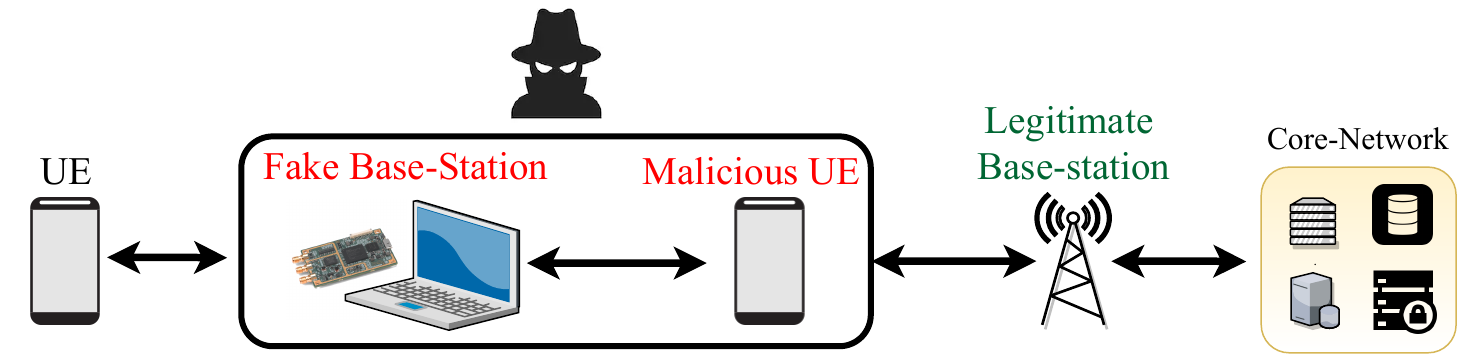}
        \caption{Setup for man-in-the-middle attacks.}
        \label{FBS_b}
    \end{subfigure}
        \caption{Common FBS configurations for carrying out attacks.}
    \end{figure}

To mitigate fake base station attacks, cellular protocol specifications have introduced enhancements across multiple generations, including 3G, 4G LTE, and 5G. A notable advancement in 5G is the introduction of the Subscription Concealed Identifier (SUCI), which protects against IMSI catchers. SUCI is an encrypted identifier used in the UE registration procedure, designed to conceal the Subscription Permanent Identifier (SUPI)---commonly known as the IMSI (International Mobile Subscriber Identity) in earlier generations. By rotating encryption keys, SUCI can change over time, making it significantly harder for IMSI catchers to track users. However, SUCIs are still transmitted over-the-air, which means an IMSI catcher could temporarily track a user if the SUCI value remains unchanged. Moreover, some network operators do not enable SUCI encryption in their systems, further undermining its effectiveness~\cite{nie2022measuring}.

\begin{table*}[ht]
\setlength\tabcolsep{4pt}
\renewcommand{\arraystretch}{0.6}
\fontsize{8}{6}\selectfont
\newcolumntype{P}[1]{>{\centering\arraybackslash}p{#1}}
\centering
\begin{tabular} {| c | P{2.8cm} | P{7.6cm} |}
\hline
\textbf{Attack} & \textbf{Attack Category} & \textbf{Impact} \\ \hline

Send RRC or NAS Reject messages \cite{shaik2015practical, hussain2018lteinspector, hussain20195greasoner, 3GPP:33.809} & DoS; Downgrade; Battery Depletion; & Denial of services; force UE to downgrade to older radio technology; Increase in power consumption for UE \\ \hline

Replay \texttt{RRC\_Resume\_Request}~\cite{3GPP:33.809} & DoS & Denial of services \\ \hline

\texttt{Authentication\_Request} with separation bit 0~\cite{rashid2024state} & DoS & Denial of services \\ \hline

Manipulate Self Organizing Networks (SON)~\cite{shaik2018impact} & DoS; Battery Depletion & Call dropping; Increase in power consumption for UE; Increased handovers and signaling load; Legitimate base station blacklisted \\ \hline

Modify \texttt{UE\_Capability\_Information}~\cite{shaik2019new} & Downgrade & Denial of some services; Lower data rate; Downgrade to 2G/3G \\ \hline

Authentication Relay Attack~\cite{hussain2018lteinspector} & Information Leak; DoS & Complete or selective DoS; Location history poisoning; Network profiling \\ \hline

5G AKA Bypass~\cite{5gbasechecker} & Information Leak & Monitor and manipulate user traffic; Provide Internet access; Phishing \\ \hline

NAS \texttt{Security\_Mode\_Command} replay~\cite{5gbasechecker} & Location Tracking; Fingerprinting & Check if the UE in the range of base station; Fingerprinting UE model \\ \hline

SUCI-Catcher~\cite{chlosta20215g} & Location Tracking; Information Leak &  Obtain user identifier; Track user movement\\ \hline

IMEI-Catcher~\cite{park2022doltest} & Location Tracking; Information Leak &  Obtain user identifier; Track user movement \\ \hline

Lullaby~\cite{hussain20195greasoner} & DoS; Battery Depletion &  Move UE to idle state; Increase in power consumption for UE \\ \hline

NAS counter reset~\cite{hussain20195greasoner} & DoS; Battery Depletion & Force UE reconnect; Increase in power consumption for UE \\ \hline

Authentication Sync Failure Attack~\cite{hussain20195greasoner} & DoS; Battery Depletion & Force UE reconnect; Increase in power consumption for UE \\ \hline

\texttt{Counter\_Check} Fingerprinting~\cite{park2022doltest, 5gbasechecker} & Fingerprinting & Fingerprinting UE model \\ \hline

\texttt{Authentication\_Request} Fingerprinting~\cite{rashid2024state} & Fingerprinting & Fingerprinting UE model \\ \hline

RRC \texttt{Security\_Mode\_Command} Bypass~\cite{kim2019touching} & Information leak & Monitor and manipulate user traffic  \\ \hline
\end{tabular}

\caption{Attacks enabled by fake base stations in 4G \& 5G cellular networks. }
\label{fbs_impact}
\end{table*}

Despite some advancements, these defenses fail to address the root cause of fake base station attacks: the lack of authentication for base stations during the connection bootstrapping process. Consequently, fake base station attacks remain feasible, even in 5G networks. To better understand the causes and implications of these attacks, we analyze fake base station attacks found in recent studies. A summary of these attacks and their impact is provided in Table~\ref{fbs_impact}.

\noindent \textbf{Denial-of-Service.}
Denial-of-Service (DoS) attacks are among the most prevalent tactics used by an FBS (Fake Base Station) attacker. By sending a reject message as defined in the specifications \cite{shaik2015practical, hussain2018lteinspector, hussain20195greasoner, 3GPP:33.809, rashid2024state}, the attacker can prevent the User Equipment (UE) from connecting to the network. Additionally, the attacker can manipulate the order or fields of control-plane messages \cite{3GPP:33.809, shaik2018impact, hussain2018lteinspector, hussain20195greasoner} to achieve a similar effect. As a result, the affected user cannot connect to a legitimate base station, rendering them unable to receive SMS, phone calls, or access the Internet. 

\noindent \textbf{Battery Depletion.}
Most cellular devices rely on battery power, making them vulnerable to energy-draining attacks. An FBS attacker can disrupt the UE's connection to the base station, forcing it into a reconnection loop that rapidly depletes the battery \cite{shaik2015practical, hussain2018lteinspector, hussain20195greasoner, 3GPP:33.809, shaik2018impact}. The frequent cell selection and registration procedures consume significant power, ultimately preventing the user from operating their device.

\noindent \textbf{Downgrade.}
In a downgrade, or bidding-down attack, the attacker forces the UE to connect to an older radio generation (e.g., 2G, 3G, LTE). This can be achieved by leveraging specific reject messages (e.g., \texttt{5GS\_Services\_Not\_Allowed}) \cite{shaik2015practical, hussain2018lteinspector, hussain20195greasoner, 3GPP:33.809} or manipulating fields in control-plane messages \cite{shaik2018impact}. Consequently, the user loses the benefits of modern radio technologies, such as faster speeds and stronger cryptographic protections. Since older generations like 2G and 3G employ weaker cryptography, attackers can exploit these connections to send fake SMS or execute other malicious activities.

\noindent \textbf{Information Leak.}
Certain FBS attacks can lead to sensitive information leakage \cite{hussain2018lteinspector, 5gbasechecker, chlosta20215g, park2022doltest, kim2019touching}. For instance, SUCI-Catcher \cite{chlosta20215g} and IMEI-Catcher \cite{park2022doltest} get the identifier of the user or device and can track the user's location. The Authentication Relay Attack \cite{hussain2018lteinspector}, 5G AKA Bypass \cite{5gbasechecker}, and RRC \texttt{Security\_Mode\_Command} Bypass attack can relay or bypass the authentication procedure between the base station and the UE, which causes the user traffic unencrypted. The attacker can monitor and manipulate the traffic and even provide Internet access to the user \cite{5gbasechecker}. Thus, the attacker can hijack users into phishing websites and perform more complicated attacks. 

\noindent \textbf{Fingerprinting.}
By analyzing UE responses to identical messages, attackers can infer device characteristics such as the baseband vendor or software version \cite{5gbasechecker, park2022doltest, rashid2024state}. With this information, attackers can execute targeted attacks tailored to specific devices or software implementations. 

\noindent \textbf{Location Tracking.}
Many attacks \cite{5gbasechecker, chlosta20215g, park2022doltest} enable location tracking of a specific UE. IMEI-Catcher \cite{park2022doltest} captures the permanent identifier of the device, IMEI (International Mobile Equipment Identity), and SUCI-Catcher \cite{chlosta20215g} captures SUCI, which is an identifier of the user. With an FBS network, the attacker can track user location if the same identifier is used elsewhere. Additionally, the NAS \texttt{Security\_Mode\_Command} replay attack \cite{5gbasechecker, hussain2018lteinspector} can verify the presence of a UE within the attacker's range by replaying previously successful \texttt{Security\_Mode\_Command} messages.

To mitigate these threats, ensuring UE authentication of base stations is crucial. This paper aims to address this gap and propose solutions to enhance security against FBS attacks.

\section{Overview of Our Solution}
\label{Overview}
\noindent \textbf{Adversary Model.}
We consider a Dolev-Yao adversary model~\cite{dolev1983security} in which the adversary can drop, modify, inject, or eavesdrop on messages sent by legitimate participants over the public radio channel. 
According to this model, the adversary is capable of setting up fake base stations and emitting unauthenticated broadcast messages with a higher signal strength than the legitimate base stations. 
However, the adversary cannot physically access and tamper the legitimate base stations, cellular devices, or core network components, and cannot access the secret keys or other sensitive information stored in a target cellular device's USIM or base stations.

\noindent \textbf{Scope of our Solution.}
Our solution allows cellular devices to reliably authenticate a base station before establishing a connection
by ensuring the authenticity of the public broadcast messages. We do not consider passive attacks caused by adversaries eavesdropping the traffic between the target device and legitimate base stations over the public radio channel. We also do not consider DoS attacks using a wireless jammer operating at the physical layer. Finally, our solution is envisioned for 5G cellular networks but can be extended to 4G LTE, 3G, and 2G networks with minimal modifications. 

\noindent \textbf{Our Authentication Protocol.}
Our protocol allows a UE to verify the identity of the base station it is connecting to and validate the authenticity of \texttt{System Information} messages being sent by the base station. Our protocol is based on an IBS scheme (details in Section~\ref{Crypto Scheme}) and organized according to a 2-layered system consisting of: \ding{182} the core-PKG (hosted by the 5GC),  and \ding{183} base stations. We provide a high-level overview of our authentication protocol below (see Section~\ref{Detailed Design} for further details). 

The core-PKG, co-located with the Authentication Server Function (AUSF) in the core network, is responsible for generating the public-private key pairs for all the AMFs deployed for a particular network operator. 
At first, core-PKG generates its public-private key pair $[sk_{PKG}, PK_{PKG}]$ during the initial setup phase. The $PK_{PKG}$ is then provisioned in the USIM of all UEs for that particular network operator. 

The base stations controlled by that network operator periodically send a key generation request to their network operator's core-PKG. The base stations send their cell IDs, i.e., \texttt{NRCell\_ID} to the core-PKG and receive a public-private key pair [$sk_{BS}$, $PK_{BS}$] and the base station identity $U_{BS}$. The $U_{BS}$ is a concatenation of \texttt{NRCell\_ID} and the expiration timestamp of the particular key pair.

The base stations use their assigned private key $sk_{BS}$ to sign the \texttt{System Information} broadcast messages using the \scheme{} (Section~\ref{Crypto Scheme}) and generate a signature $sig_{SIB1}$. The base stations include $sig_{SIB1}$, $PK_{BS}$, and $U_{BS}$ in SIB1 message.
The UE uses this information to verify $sig_{SIB1}$. The UE first verifies if the key $PK_{BS}$ is expired by checking the expiry timestamps embedded in the $U_{BS}$. If the timestamp is not expired, the UE verifies the signature of the \texttt{System Information} message. If this verification step is successful, the UE connects to the base station.

\section{New Efficient 2-Layer Identity-Based Signature Scheme (\scheme{})}\label{sec:scheme}
\label{Crypto Scheme}
In this section, we provide the general definition and discussion of our newly proposed scheme, the Efficient 2-layer Identity-Based Signature Scheme (\scheme{}).

In PKI-based schemes, the security of cryptographic primitives (e.g., digital signatures) hinges on the authenticity of public keys. In conventional systems, this is achieved by digital certificates or certificate chains. 
For instance, to verify a digital signature, the verifier must confirm the authenticity of the signer's public key by ensuring the validity of the associated certificates. However, the communication and computation overhead introduced by certificates might not be tolerable in some applications (mobile devices operating in low-bandwidth environments). To address this, in identity-based cryptography, the user's public key is derived from their publicly available information (e.g., IP address).  
Existing efficient identity-based signature schemes (e.g., \cite{singla2021look}) are primarily based on the Schnorr signature \cite{Schnorr91}. 
Despite their elegant design, their inherent design and the key generation process (which derives keys from the Schnorr signature) can result in less efficient verification algorithms. 

With the efficiency and security requirements of 5G networks in mind, we present a new Efficient 2-layer Identity-Based Signature (\scheme{}).  \scheme{}, presented in  Algorithm \ref{alg:IBS}, offers highly efficient signing and verification, ensures high resiliency by avoiding a single point of failure, and supports fine-grained lawful interception.  
This is achieved by deriving a new identity-based signature from the highly efficient certificate-based scheme, ARIS \cite{ARIS}. The efficiency of ARIS is due to the ability to convert costly exponentiation operations to a few point additions by utilizing the homomorphic property of the underlying one-way function. This results in computation efficiency in both signing and verification algorithms. 

As depicted in Algorithm \ref{alg:IBS}, after the parameter selection (similar to \cite{ARIS,Tachyon}), the Setup algorithm computes the $t$ public key elements $mpk=\{Z_i\}_{i=1}^t$ and publishes the public parameters. We note that parameters $t$ and $k$ are related to the k-combinatorial problem \cite{Tachyon,ARIS}, i.e., ${t \choose k}\geq 2^\kappa $ (for security parameter $\kappa$)  and play an important role in providing storage and computation overhead trade-off.  For instance, a larger $t$ results in larger keys but more efficient signing and verification, as it allows for a smaller $k$.   During the extract algorithm, by harnessing the scheme in \cite{ARIS}, the PKG computes the user's key pair ($sk_U,{C}_U$) based on the provided identity $U$.   
After computing the user keys and considering their structure, the signing and verification processes can be carried out in a manner similar to that in \cite{Schnorr91}.
 
\newcommand{\algrule}[1][.2pt]{\par\vskip.3\baselineskip\hrule height #1\par\vskip.3\baselineskip}

\begin{algorithm}\caption{$\scheme{}$}\label{alg:IBS}
\small
$(msk,params)\gets\setup(1^\kappa)$
\algrule[0.5pt]
\begin{algorithmic}[1]
 
\item Given $\kappa$, select $p,q$, $msk \Ra \mathbb{Z}_p$ and  $t,k\Ra \mathbb{N}$ where   ${t \choose k}\geq 2^\kappa $
\item  Compute $z_i \gets \PRF_{msk}(i)$  and $Z_i \gets z_iP \mod q$ $\mathbf{for}$  $i= \{1,\dots,t\}$ and set $\mathbf{Z}\gets \{Z_i\}_{i=1}^t$
\item  Output $msk$ and $params=(mpk,p,q,t,k)$, where $mpk=\mathbf{Z}$
 
\end{algorithmic}
\algrule[0.5pt]

 $(\sk_U,{C}_U)\gets\sgnextract(msk,U)$
\algrule[0.5pt]

\begin{algorithmic}[1]

\item Compute $u \gets \PRF_{msk}(U)$ and ${C}_{U}\gets u P \mod q$ 
\item Compute $\{j_1\dots, j_k\}\gets \h_1(U,C_{U})$ where each $|j_i| = |t|$

 \item Compute $x_{U}\gets   \sum_{i=1}^{k} z_{j_i} + u \mod p$  
  \item Output $(\sk_U  = x_{U}$, ${C}_U$) 
\end{algorithmic}
\algrule[0.5pt]

$\sigma_{m,U}\gets\sign(m,\sk_U)$
\algrule[0.5pt]
\begin{algorithmic}[1]
\item Select $r \Ra \mathbb{Z}_p$ and compute $  h\gets \h_2(m,rP \mod q) $
\item Compute $s\gets r - h \times \sk_U$
\item Outputs $\sigma_{m,U} = (s, h)$

\end{algorithmic}

\algrule[0.5pt]
$\{\text{valid},\text{invalid}\}\gets\verify(m,\sigma_{m,U},U,{C}_U,mpk)$
\algrule[0.5pt]
\begin{algorithmic}[1]

\item Compute $\{j_1,\dots,j_{k}\} \gets \h_1(U,C_U)$ 
\item $R'\gets sP+h(\sum_{i=1}^k \mathbf{Z}[j_i] \mod q +C_U)$
\item Output `valid' \textbf{if} $h=\h_2(m,R')$, \textbf{else} output `invalid'

\end{algorithmic}

\end{algorithm}

\subsection{Robust Security and fine-grained lawful interception} \label{sec:lawful}
Identity-based systems solely rely on the PKG to generate the keys for all users, creating a single point of failure. Thus, if the PKG is compromised, the entire system's security is at risk. To address this vulnerability, we leverage the additive property of the key generation algorithms of \scheme{} to propose new key generation methods (Algorithm \ref{alg:IBS_lawful}). This is achieved by dividing the signer's key into two components, $u_1$ and $z_U$. During the new key generation process, the signer selects $u_1$ and computes its commitment $Q_u$.
The PKG then computes the other secret component, i.e., $z_U$, by deriving an ARIS signature (see $\sgnextract(\cdot)$ in Algorithm \ref{alg:IBS_lawful}) on user identity $U$ and $Q_u$ (implicit certification).  After receiving $z_U$, by leveraging the additive property of \scheme{}, the signer computes the \emph{final} secret key as $x_U\gets u_1+z_U \mod p$. In this case, even if the PKG is compromised, the user secret key remains secure since the adversary needs knowledge of $ u_1$ to compute it. 

With this improvement, the new key generation algorithm can also enable a fine-grained lawful interception by allowing the signer to reissue its key by running $(u_1,Q_U)\gets \userkg(params)$  and requesting a new $x_U$ from the PKG. This requires including a sequence number $t$, supplied by the user $U$, in the input of the hash function $\h_1(\cdot)$ during the $\sgnextract(\cdot)$ algorithm. Using the sequence number in key generation prevents the misuse of the old keys and provides an efficient approach for fine-grained control of the key lifespan. 

\begin{algorithm}\caption{Key Generation for robust security and fine-grained lawful interception}\label{alg:IBS_lawful}
\small
$(u_1,Q_U)\gets\userkg(params)$
\algrule[0.5pt]

\begin{algorithmic}[1]
 
\item Compute $u_1 \gets  \ZZ_p$
\item Compute $Q_U \gets u_1 P \mod q $
\item Output $(u_1,Q_U)$
 
\end{algorithmic}
\algrule[0.5pt]

 $(x_U,{C}_U)\gets\sgnextract(msk,Q_U)$
\algrule[0.5pt]

\begin{algorithmic}[1]

\item Compute $u_2 \gets \PRF_{msk}(U)$ and ${B}_{U}\gets u_2 P \mod q$ 
\item Compute $C_U \gets Q_U+B_U \mod q$
\item Compute  $\{j_1\dots, j_k\}\gets \h_1(U,C_{U})$ where each $|j_i| = |t|$

 \item Compute $z_{U}\gets   \sum_{i=1}^{k} z_{j_i} + u_2 \mod p$  
  \item Output $( z_{U}, B_U,{C}_U$) 
\end{algorithmic}
\algrule[0.5pt]

$\sk\gets\cmpkey(u_1,z_U)$
\algrule[0.5pt]

\begin{algorithmic}[1]
 
\item Compute $x_U\gets u_1+z_U \mod p$
\item Output ($x_U,Q_U$)
 
\end{algorithmic}

\end{algorithm}

\subsection{Security Analysis}
\begin{theorem}
    In the random oracle model, if an adversary \A~can break the scheme proposed in Algorithm \ref{alg:IBS}, in the sense of Definition \ref{def:eucma}, then one can construct another algorithm \C~that runs the adversary and \A~as a subroutine and can solve an instance of the ECDL problem in Definition \ref{def:ecdl}.
\end{theorem}
 
\begin{proof}

    Given  $X\gets \EC$ as an instance of the ECDL problem, \C~works as follows to find a solution $z^*\gets \ZZ_p$, such that $z^*P=Z^* \mod q$. 
    
    \noindent\emph{Setup:} \C~keeps two lists ($L_1, L_2$) to keep track of the output of the random oracles $\h_1(\cdot)$ and $\h_2(\cdot)$  and lists $L_\sigma$ and $L_U$ to keep track of the messages submitted to the sign and corrupt oracles, respectively. \C~sets up the following random oracles to handle queries to hash functions. 
    \begin{itemize}
        \item $\alpha_1\gets \h_1$-$\mathtt{sim}(U,C_{U},L_1)$: If the input (i.e., $U,C_{U}$) already exists, it returns the corresponding $\alpha_1$, else, it returns  $\alpha_1 \Ra \{0,1\}^{k|t|}$ and stores $(U,C_{U},\alpha_1)$   in $L_1$. 
        
        \item $\alpha_2\gets \h_2$-$\mathtt{sim}(m, R, L_2)$: If the input $(m,R)$ already exists, it returns the corresponding $\alpha_2$, else, it returns  $\alpha_2 \Ra  \ZZ_p$ and stores $(m, R,\alpha_2)$   in $L_2$. 
            
            \end{itemize}
     Next, \C~selects a target index $j^*\gets \{1,\dots,t\} $ and sets the target $mpk$ element $ Z_{j^*} =  Z^*$. Then it selects $z_i \Ra \ZZ_p$ and compute  $Z_i\gets z_i P \mod q $ where $i\in\{1,\dots,t\}$ and $i\neq j^*$ and output $mpk=(Z_i, \dots, Z_t)$.

    \noindent \emph{Queries:} 
    
    \begin{itemize}
    \item \emph{Hash queries:} Hash queries on $\h_1$, $\h_2$ and $\h_3$ will be handled by   $\h_1$-$\mathtt{sim}(\cdot)$, $\h_2$-$\mathtt{sim}(\cdot)$ and $\h_3$-$\mathtt{sim}(\cdot)$ functions defined above, respectively. 
    \item \emph{$\mathcal{O}_{Corrupt}({U})$ Queries:} Given a user $U$, if $U$ exists in $L_U$, it returns $(x_{U})$. Next, it checks $L_1$; if such $U$ exits with an index corresponding to $j^*$, it aborts. Else, it selects $u\Ra \ZZ_p$, computes $U\gets uP \mod q$. Next, it selects $j_i\Ra\{1,\dots,t\}$ for $i=\{1,\dots,k\}$ and $j_i \neq j^*$, for each $j_i$, recovers the corresponding $z_{j_i}$ and computes and returns the secret key $x_{U}\gets \sum_{i=1}^k z_{j_i}+u \mod p$. To respond to future queries, the output is stored in $L_U$. 
    \item \emph{$\mathcal{O}_{Sign}(m,U)$ Queries:} For signature queries for users $U$ where $\alpha_1 $ does not contain $j^*$, \C~can work similar to the $\sign(\cdot)$ to generate the signature. Otherwise, when $\alpha_1$ contains $j^*$, \C~uses its access to the random oracle and works similarly to Schnorr Signature to generate a valid signature for \A.
    
    \end{itemize}
    \noindent \emph{$\A$'s Forgery:}  Finally, $\A$ will output a forgery message-signature pair $(m^*,\sigma_{U^*})$. $\A$ wins the game if $\text{`valid'}\gets\scheme.\verify(m^*,\sigma_{U^*},U^*,{C}_{U^*},mpk)$ and  $m^*$ was not submitted to $\mathcal{O}_{Sign}(\cdot)$.   

\noindent \emph{Solving the hard problem:} After outputting a valid forgery by $\A$, $\C$ checks if for $U^*$  the target public key element $Z^*$ is embedded $\alpha_1$. Else, it fails. If   $\alpha_1$ indeed contains $j^*$, similar to \cite{Tachyon,ARIS}, \C~utilizes the forking lemma \cite{Bellare-Neven:2006}, to obtain a second forgery $m^*,\sigma'_{U^*}$, where with a very high probability $s^*\neq s'$ and $h^*= h' $. Then, given the results of Lemma 1 in \cite{Bellare-Neven:2006}, to solve the ECDL problem.
    
\end{proof}

\begin{cor}
 The key generation algorithm provided in Algorithm \ref{alg:IBS_lawful} offers robust security by  preventing a single point of failure inherent in identity-based schemes. 
\end{cor}
\begin{proof}
      In the original scheme (i.e., Algorithm \ref{alg:IBS}), the secret component of the user key is computed as $x_U\sum_{i=1}^{k} z_{j_i} + u \mod p$, where both $z_j$'s and $u$ are known and selected by the PKG.  
      Consequently, a compromised PKG can issue private keys on behalf of the user. In the key generation algorithm with robust security in Algorithm \ref{alg:IBS_lawful}, the secret component of user key is computed as $x_U = \sum_{i=1}^{k} z_{j_i} + u_2+u_1$ where $u_1$ is not known to the PKG. This offers the binding property by incorporating the commitment of $u_1$ as the input of the hash function $\h_1(\cdot)$. We note that the correctness of the key supplied by the PKG can be simply  verified by the user by running the verification algorithm in ARIS \cite{ARIS} on $z_U$.
\end{proof}

\section{Instantiation of \scheme{} for 5G Networks}
\label{Detailed Design}
In this section, we discuss the limitations of \texttt{Schnorr-HIBS} and introduce our new scheme. 

The preliminary version of \scheme{}, \texttt{Schnorr-HIBS} has a hierarchal design with 2 PKGs. As shown in Figure \ref{pre_protocol}, the core-PKG generates key pairs for AMFs as the first-level PKG. The AMFs, serve as second-level PKGs, generate key pairs for their corresponding base stations. Then, the base stations sign their \texttt{SIB1} messages using the received private keys and broadcast the message to the UEs. The UEs, with the $PK_{PKG}$ provisioned, need to receive all the public keys and identities from the base station and the AMF to verify the signature. 

\begin{figure*}[t]
 \centering
    \includegraphics[width=0.9\linewidth]{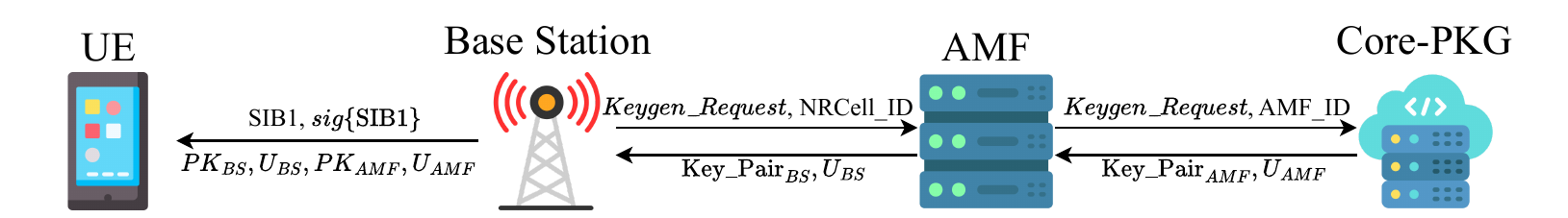}
    \caption{Instantiation of \texttt{Schnorr-HIBS}. }
        \label{pre_protocol}
\end{figure*}

\subsection{\texorpdfstring{Limitations of \texttt{Schnorr-HIBS} \cite{singla2021look}}{Limitations of Schnorr-HIBS}}
\label{limitations}

\noindent \textbf{Hierarchal Architecture. }
\texttt{Schnorr-HIBS} proposed a hierarchal architecture signature scheme, which supports multiple layers from the PKG to the verifier. In the instantiation, the core-PKG and the AMF serve as the first- and second-level PKG. However, this hierarchal architecture has several limitations: \ding{182} Using AMFs in the middle creates more complexity for the scheme. The operators need to implement and manage the crypto functionalities in AMF, which introduces more cost. A misconfigured AMF can make the system unavailable and even leak its private keys. If an attacker has control over an AMF, it can sign its own fake base stations and the user cannot detect it from the signature. \ding{183} The core-PKG is centralized, it can be configured with powerful hardware and the signature extraction is very fast. However, the AMF is distributed and can have various types of configurations. It can be difficult to ensure service quality if one AMF is serving a large amount of base stations. \ding{184} Adding additional levels in the scheme will cause more overhead on the verifier side. Because the verifier needs to get the public keys for all layers and verify the signature with more arithmetic operations. In addition, the multi-level design introduces more communication costs. 

\noindent \textbf{Lawful Interception. }
\texttt{Schnorr-HIBS} assumed that the law enforcement department could obtain key pairs for base stations from the AMF and set up their fake base stations like a normal base station. However, it did not provide a fine-grained access control mechanism to limit the usage of the key and revoke the key if needed. In that case, the law enforcement department can set up a fake base station to intercept user traffic at any time and any location with a valid key pair. This may violate the authorized scope, and give the law enforcement department more power to track users than it should have. Furthermore, if the key is leaked, the attacker can set up \textit{authenticated} fake base stations, which are considered valid by the UE. 

\noindent \textbf{End-to-end Implementation. }
\texttt{Schnorr-HIBS} did not provide an end-to-end implementation for the proposed scheme. Thus, it overlooked the challenges of deploying the proposed scheme in the real world. For example, the hierarchal architecture requires new implementation for both core-PKG and AMF, increasing the complexity and the development and maintenance costs. Also, the base station needs to send the signature of \texttt{SIB1} along with the public key and identity of both the base station and the AMF to the UE. The total overhead is 150 bytes, which takes a large amount of the available space in \texttt{SIB1} (372 bytes) \cite{3GPP:38.331} message. A commercial base station may not be able to append this information after the existing configurations. 

\subsection{Design Decisions of \scheme{}}
\label{design_decisions}
 To address these limitations, 
we outline the detailed design of our protocol and the rationale behind the design decisions. 

\noindent \textbf{2-layer Design. } 
We specify a 2-layer architecture for our protocol: the core-PKG generates the keys for base stations and the base station creates the signatures. We use a 2-layer approach instead of a hierarchical approach where a core-PKG generates keys for AMFs and the AMFs generate the signing keys for the base stations for several reasons: 
In our approach, the AMF only needs to forward the \textit{keyext\_request} and the response from core-PKG. 
In our scheme, we are aiming to make the scheme both signer and verifier efficient.

\noindent \textbf{Fine-grained Lawful Interception. } 
We design our protocol with fine-grained lawful interception in mind. Algorithm \ref{alg:IBS_lawful} introduces key generation with a sequence number, ensuring that keys generated with a previous sequence number are implicitly invalidated when a new sequence number is used. This mechanism enables seamless key revocation and can be further enhanced by integrating fine-grained access control policies directly within the core-PKG, ensuring efficient and secure key management. 

\noindent \textbf{Minimize Bytes Sent Over-The-Air.} 
To comply with the current protocol and introduce a minimum overhead while satisfying modern security requirements, we use a Schnorr signature scheme. 
Only 111 bytes (see \ref{Sec:counterparts}) are required to send over-the-air to the UE. The communication overhead is 26\% smaller than the previous scheme.  

\noindent \textbf{Choice of Messages to Sign.} \texttt{System information} messages are broadcast periodically by the base stations to allow UEs to initiate a connection to them. \texttt{System Information} messages are divided into a \texttt{Master Information Block (MIB)} and multiple \texttt{System Information Block (SIB)} messages~\cite{3GPP:38.331}. \texttt{MIB} includes the basic parameters required by the UE to acquire the \texttt{SIB1} message. The \texttt{SIB1} message is the most important \texttt{System Information} message and contains the base station selection parameters, scheduling info for the rest of the SIB messages, whether one or more SIB messages are only provided on-demand, and configuration needed by the UE to perform the system information request.  
Since the MIB and SIB1 messages are two messages required for a UE to connect to a base station, our protocol signs the two messages together and provides the signature in the SIB1 message. After the UE receives the SIB1 message, it is able to authenticate the base stations. 

\noindent \textbf{Construction of Identities.} Our protocol requires assigning IDs to the base stations. We utilize the IDs for the dual purpose of uniquely identifying the base stations as well as for communicating the validity period of their signing keys. 
For $U_{BS}$ we use a concatenation of \texttt{NRCell\_ID}~\cite{3GPP:29.571} and an expiry timestamp. \texttt{NRCell\_ID} is a string of size 36 bits and uniquely identifies a base station for a particular mobile network operator. Each expiry timestamp is 32 bits long. Therefore, $U_{BS}$ can be a maximum of 9 bytes. 

\noindent \textbf{Validity period of the keys.}
\label{validity} Instead of using complex key revocation techniques, we assign different validity periods to each generated keypair after which the keys would need to be refreshed. For the core-PKG, we create the key-pair with a 1 year validity period by default as it needs to be installed inside the UE's USIM, and requires a confidentiality and integrity-protected channel to be updated. The core-PKG needs to be physically secured and protected so that its private key is not leaked. 
For the base stations, we generate a key pair valid for only 10 minutes. base stations are located around the world in physically insecure areas. Therefore, it may be easier for the attacker to compromise them. A validity period of only 10 minutes minimizes the period during which an attacker can launch attacks, even if it obtains a base station's private key. These validity periods are recorded in the expiry timestamps in the $U_{BS}$, as well as in the UE's USIM for the core-PKG. These are the default validity periods and can be changed by the network operators when required. Since our key generation is efficient (1000 keys per 5.5 milliseconds), its impact on core-PKG is negligible. 

\begin{figure*}[t]
 \centering
    \includegraphics[width=0.9\linewidth]{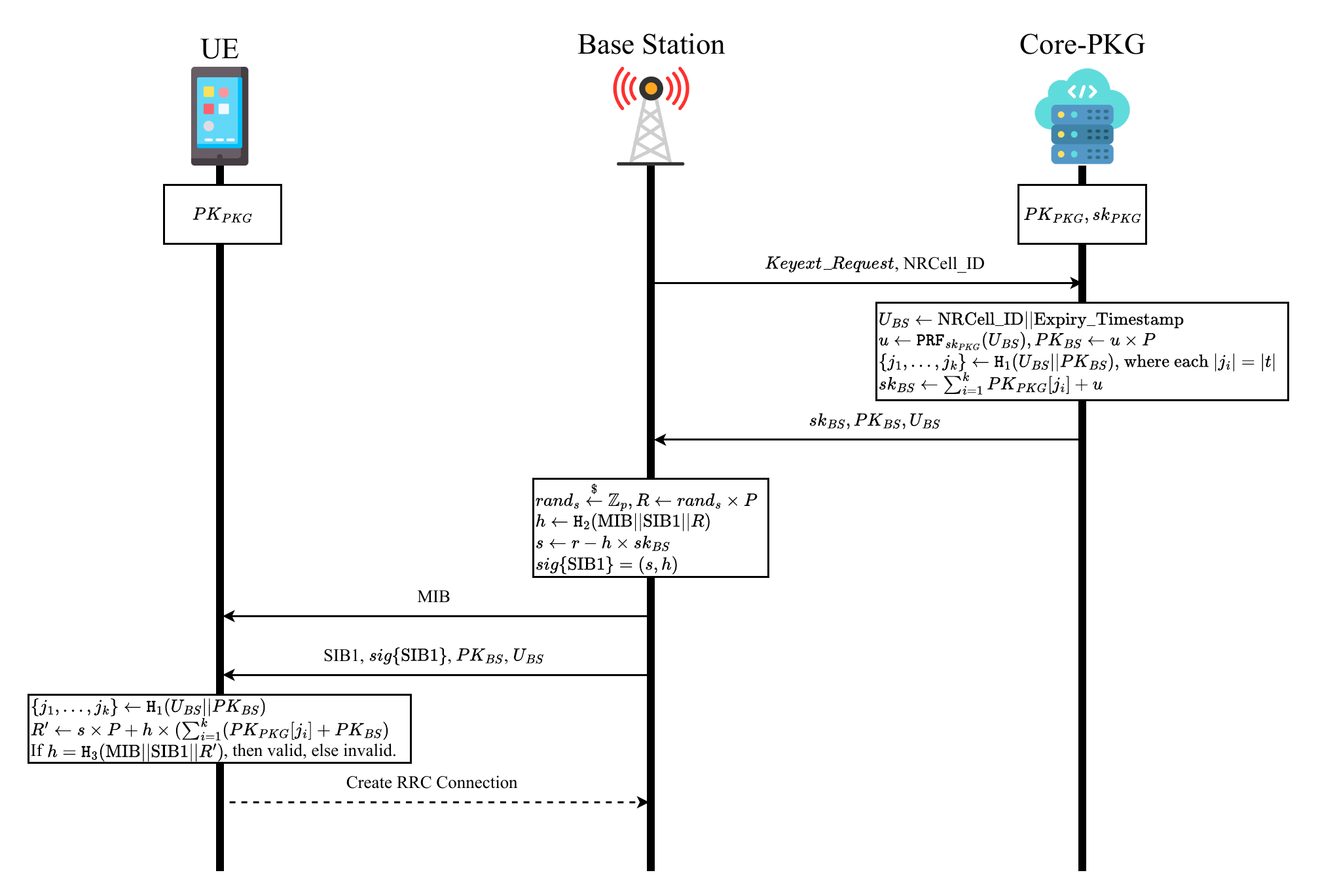}
    \caption{Our protocol for authenticating 5G cellular base stations.}
        \label{protocol}
\end{figure*}

\subsection{Protocol Description}\label{protocolDescription}
We now detail our authentication protocol steps. We abstract some cryptographic details for readability. 
For instance, we do not explicitly mention the mod operation, but all the operations in $E(\mathbb{F}_p)$ are executed in mod~$p$ and operations in $\mathbb{Z}_q$ are in mod~$q$. 
Figure~\ref{protocol} gives a graphical representation of our protocol for a 5G scenario. 

\subsubsection{Initialization phase for the core-PKG}
\noindent The core-PKG generates the public system parameters and its own public-private key pair during the initialization phase. From $sk_{PKG}$, core-PKG generates $t$ sets of public-private key pairs. $PK_{PKG}$ contains all $t$ public keys generated. This phase is executed at the beginning of the 5GC deployment. The default validity period of the core-PKG's keys is 1 year. The public key of the core-PKG along with its expiry date is installed in the USIM of all UEs during initial registration. The core-PKG's public key installed in the USIM has to be replaced, whenever the core-PKG refreshes its keys. This can be done using the confidentiality and integrity-protected channel created between the AMF and the UE after mutual authentication. The core-PKG uses the \texttt{Setup} step from Algorithm~\ref{alg:IBS} to generate its key pair
\{$sk_{PKG},\ PK_{PKG}$\}.

\begin{gather*}
sk_{PKG} \gets \mathbb{Z}_p \\
z_i \gets \PRF_{sk_{PKG}}(i), Z_i \gets z_i \times P \textbf{ for } i\in \{1,\dots,t\} \\
PK_{PKG} \gets \{Z_i\}_{i=1}^t
\end{gather*}

\subsubsection{Key extraction for the base station}
\noindent Base stations send a key extraction request with their \texttt{NRCell\_ID} to the core-PKG through their serving AMFs. The core-PKG concatenates the received \texttt{NRCell\_ID} with an expiry timestamp for the key being generated as $U_{BS}$. From $U_{BS}$, the core-PKG follows the \texttt{Extact} step from Algorithm~\ref{alg:IBS} to generate $sk_{BS}$ and $PK_{BS}$. The base stations have to periodically refresh their key-pair by sending the key-generation request to the core-PKG when nearing the key-pair expiration. 

\begin{gather*}
U_{BS} \gets \text{NRCell\_ID} || \text{Expiry\_Timestamp} \\
u \gets \mathtt{PRF}_{sk_{PKG}}(U_{BS}), PK_{BS} \gets u \times P \\
\{j_1, ..., j_k\} \gets \mathtt{H}_1(U_{BS}||PK_{BS}) \text{, where each } |j_i|=|t| \\
sk_{BS} \gets \sum^k_{i=1}PK_{PKG}[j_i] + u
\end{gather*}

\subsubsection{Signing phase at the base station} 
The base stations sign the MIB and SIB1 message via \texttt{Sign} step of Algorithm~\ref{alg:IBS} and generate the signature $sig_{SIB1}$. They attach the $sig_{SIB1}$, $PK_{BS}$, and $U_{BS}$ along with the SIB1 message broadcast. Before signing, the base station needs to ensure that their own keys have not expired. \begin{gather*}
rand_s \stackrel{\$}{\leftarrow} \mathbb{Z}_p, R \gets rand_s \times P\\
h \gets \mathtt{H}_2(\text{MIB}||\text{SIB1}||R)\\
s \gets r - h \times sk_{BS}
\end{gather*}
where $\langle s,\ h \rangle$ is the signature.

\subsubsection{Verification phase at the UE} 
The UE uses the $U_{BS}$ and the $PK_{BS}$ sent by the base station attached to the SIB1 message to verify $sig_{SIB1}$. The UE first verifies that the key $PK_{BS}$ is not expired by looking at the expiry timestamps embedded in the $U_{BS}$. If the timestamps have not expired, the UE computes the indices to select the corresponding public keys from $PK_{PKG}$ and then verifies the signature $sig_{SIB1}$.
For verification, the UE uses the public keys of core-PKG and the base station: 
\begin{gather*}
\{j_1,...,j_{k}\} \gets \mathtt{H}_1(U_{BS}||PK_{BS})\\
R' \gets s \times P + h \times (\sum^{k}_{i=1}(PK_{PKG}[j_i] + PK_{BS})\\
\text{If } h = \mathtt{H}_3(\text{MIB}||\text{SIB1}||R') \text{, then valid, else invalid.}
\end{gather*}
\noindent \textbf{Authentication failure action.} In case of authentication failure or the absence of authentication capabilities at the base station, the UE does not connect to the base station and keeps searching for other base stations available in the area. If there are no available base stations that can be authenticated, the UE can connect to an unauthenticated base station or keep looking for a base station that can be authenticated. We propose this to be a UE-specific choice, which can be configured depending on the mobile user's security/connectivity needs. If the UE decides to connect to an unauthenticated base station, it keeps checking the \texttt{System Information} messages to find a base station that can be authenticated.

\subsection{Handling Roaming Scenario}
\label{roaming}
Roaming services enable a UE to connect to base stations operated by a different network operator. Since each operator manages its own core-PKG, the UE must first obtain the public key of the roaming operator. The UE's primary operator can sign the roaming operator’s public key and provision it through non-3GPP access networks, such as Wi-Fi.

\subsection{Protection Against Relay Attacks}
Our authentication protocol protects against fake base stations by allowing UEs to authenticate \texttt{System Information} messages. However, it is vulnerable to relay attacks, where an adversary retransmits these messages with a stronger signal, tricking the UE into connecting to a fake base station. Distance-bounding protocols \cite{rasmussen2010realization, tippenhauer2015uwb, durholz2011formal} could prevent these attacks but would require major changes to cellular protocols. An alternative is to time-bound \texttt{System Information} message validity by estimating the time for an adversary to intercept and retransmit messages, though this doesn't account for base station frequency or coverage area. The 5G base stations can vary in configurations and use different frequencies to cover different ranges \cite{3GPP:38.104}, making the use of a fixed transmission time less practical. 

To protect against relay attacks in all scenarios, we propose time-bounding the \texttt{System Information} message signatures based on the base station’s configuration. The validity period, denoted by $\mathsf{\bigtriangleup t}$, is calculated from the configuration-specific transmission delay ($\mathsf{\bigtriangleup t_{conf}}$) and cryptographic signature delay ($\mathsf{\bigtriangleup t_{sign}}$). $\mathsf{\bigtriangleup t_{conf}}$ is configured by the operators and can be derived from a lookup table stored securely in the base station's memory. $\mathsf{\bigtriangleup t_{sign}}$ varies with the cryptographic scheme used. The base station signs the message with a timestamp $\mathsf{T_{sign}}$ and the validity period. The UE checks the validity by verifying if $\mathsf{T_{current} < T_{sign} + \bigtriangleup t}$ when receiving the message.

\section{Evaluation}
\label{Evaluation}
Here we evaluate \scheme{} with other state-of-the-art authentication schemes and report the experiment results. 
Furthermore, to demonstrate the practicality of our scheme, we incorporate our scheme into the open-source 5G implementation. The experiment results show our scheme has a negligible performance impact on the base stations and the UEs. 
The implementations used in the evaluation are available at \cite{E2IBS-github}. 
\subsection{Testbed Setup}

\noindent \textbf{Hardware and software components.}
We evaluate the efficacy of our scheme on a desktop machine with Intel Core i9-14900k and 64 GB DDR5 RAM, running Ubuntu 22.04. We use the PBC-library~\cite{pbc} for implementing pairing-based schemes and the FourQlib~\cite{fourqlib} for FourQ elliptic curves.

For the 5G testbed, we implemented our scheme on top of the open-source 5G stack implementation, OpenAirInterface~\cite{openAirInterface}. OpenAirInterface (OAI) has 5G Core, base station, and UE implementations. We modified the components to support our scheme. Since significant changes are required for a normal UE to authenticate the base station, we can only use our modified OAI-UE to conduct the evaluation. We plan to open-source our modified cellular stack to help further research in this field. We use two \textit{USRP B210}~\cite{USRP_B210} software-defined-radio (SDR) boards connected to the desktop machine mentioned above. One \textit{USRP B210} serves as the 5G base station and another serves as the 5G UE.

\subsection{Evaluation Results}
\noindent \textbf{Signature Schemes.}
\label{Sec:counterparts}
We consider 6 signature schemes for qualitative and quantitative comparisons with our scheme. Following the recommendations of 3GPP~\cite{3GPP:33.809}, we evaluate identity-based signature schemes BLS~\cite{boneh2001short}
We believe that BLS has been mentioned as an identity-based scheme by mistake. While one can devise an identity-based signature from BLS, such schemes are deemed to be very expensive (see Section 3.3 in~\cite{DBLP:series/ciss/KiltzN09}).  
3GPP also recommends the ECCSI scheme from RFC-6507~\cite{rfc6507}, which is based on the improved schemes proposed by Arazi et al.~\cite{Arazi-SelfCertified,araz2006load}. However, to the best of our knowledge, there is no provable security argument for these schemes, and the earlier versions of such schemes are insecure~\cite{Arazi-Attacked}. We thus omit this scheme from our comparisons. We include ECDSA~\cite{johnson2001elliptic} in our comparisons because of its wide adoption. We also include SCRA-BGLS~\cite{yavuz2017real}, since it is a recent proposal for base station authentication in 5G~\cite{hussain2019insecure}. 
In addition, we evaluate two fast authentication schemes based on the FourQlib, SchnorrQ \cite{SchnorrQ} and ARIS \cite{ARIS}. We also evaluate our preliminary scheme, \texttt{Schnorr-HIBS} \cite{singla2021look}. 
We have implemented the pairing-based schemes with the PBC-library~\cite{pbc} curve d224 which provides 112 bits of symmetric key security. These schemes are even slower for the 128-bit security level. All other schemes provide 128-bit symmetric key security according to NIST recommendations~\cite{keylength}. 
Table~\ref{comparison_schemes} summarizes our comparisons.

\begin{table*}[t]
\centering 
\renewcommand{\arraystretch}{}
			\begin{tabular}{|c||c|c|c|c|c|c|c|c|c|c|}
\hline
\multirow{2}{*}{\textbf{Scheme}} & \multicolumn{2}{c|}{\textbf{Sign}} & \multicolumn{2}{c|}{\textbf{Verify}} & \multirow{2}{*}{\begin{tabular}[c]{@{}c@{}}\textbf{Signature}\\ (B)\end{tabular}} & \multirow{2}{*}{\begin{tabular}[c]{@{}c@{}}\textbf{PK}\\ (B)\end{tabular}} & \multirow{2}{*}{\begin{tabular}[c]{@{}c@{}}\textbf{Crypto E2E }\\ \textbf{Delay} ($\mu$s)\end{tabular}} & \multirow{2}{*}{\textbf{System}} & \multirow{2}{*}{\begin{tabular}[c]{@{}c@{}}\textbf{Scheme} \\ \textbf{Type}\end{tabular}} & \multirow{2}{*}{\begin{tabular}[c]{@{}c@{}}\textbf{Lawful}\\ \textbf{Interception} \end{tabular}} \\ \cline{2-5}
 & \multicolumn{1}{l|}{$\mu$s} & \multicolumn{1}{l|}{Cycles} & \multicolumn{1}{l|}{$\mu$s} & \multicolumn{1}{l|}{Cycles} &  &  &  &  & & \\ \hline

ECDSA-256~\cite{johnson2001elliptic} & 521.14 & 1.662 & 172.78 & 0.550 & 64 & 32 & 693.92 & Flat & CB & No \\ \hline
BLS~\cite{boneh2001short}$\dagger$ & 1658.78 & 5.287 & 4957.02 & 15.799 & 48 & 96 & 6615.80 & Flat & CB & No \\ \hline
SCRA-BGLS~\cite{yavuz2017real} & 79.75 & 0.254 & 50265.27 & 160.208 & 29 & 85 & 50345.02 & Flat & CB & No \\ \hline
SchnorrQ~\cite{SchnorrQ} & 6.29 & 0.020 & 10.36 & 0.033 & 64 & 32 & 16.65 & Flat & CB & No \\ \hline
ARIS~\cite{ARIS} & 8.02 & 0.026 & 6.99 & 0.022 & 64 & t*32 & 15.01 & Flat & CB & No \\ \hline
\texttt{Schnorr-HIBS}~\cite{singla2021look} & 5.93 & 0.019 & 30.23 & 0.096 & 64 & 32 & 36.16 & Hierarchical & IDB & No \\ \hline
\hline
\scheme{} & 6.04 & 0.019 & 15.45 & 0.049 & 64 & 32 & 21.48 & Hierarchical & IDB & Yes \\ \hline
\end{tabular}

\begin{tablenotes}[list=off,flushleft]
\footnotesize{
\item All sizes are in bytes, and all computations are in microseconds. We also represent the number of CPU cycles for computation in millions.
\textbf{Signature} and \textbf{PK}  represent the signature size and public size, respectively. \textbf{Scheme Type} indicates whether the scheme is certificate-based (CB) or identity-based (IDB). \textbf{Crypto E2E Delay} for certificate-based schemes includes the verification of the sender's public key authenticity via certificates provided by a CA. For ARIS, we use t=1024, and for \scheme{}, we use t=1024 and k=18 in our evaluation. To be favorable to certificate-based schemes, we only consider a certificate chain of size 1. We implemented the pairing-based schemes with the PBC-library~\cite{pbc} curve d224, providing 112-bit security. These schemes are even slower for the 128-bit security level. All other schemes provide 128-bit security according to NIST recommendations~\cite{keylength}. 

$\dagger$ BLS is listed as an identity-based scheme by 3GPP~\cite{3GPP:33.809} but it is certificate-based (see Section \ref{Sec:counterparts}). 
}
\end{tablenotes}

\caption{Quantitative and qualitative comparison of the candidate signature schemes for authenticating cellular base stations.}	\label{comparison_schemes} 
\end{table*}

\noindent \textbf{Quantitative comparison.}
We evaluate signing and verification costs, end-to-end cryptographic delay, and total communication overhead of the compared schemes. 

\noindent \textit{\underline{Signing cost:}} Signature generation in \scheme{} only takes 6.04 $\mu$s which is 13x faster than SCRA-BGLS and 86x faster than ECDSA. \scheme{} is only slightly slower than Schnorr-HIBS. Keeping the signing time low is critical for base stations as new SIB1 messages are transmitted with a periodicity of 160 ms. Within this 160 ms, repetitive transmission occurs every 20 ms typically. A low signing overhead should also incentivize the network operators to configure the base stations to sign all broadcast messages providing full integrity protection.

\noindent \textit{\underline{Verification cost:}} Our scheme also has the second fastest verification phase taking just 15.45 $\mu$s. Our scheme outperforms the ECDSA by 11x and \texttt{Schnorr-HIBS} by 2x. Only ARIS and SchnorrQ are faster than \scheme{}. However, in the 5G setup, we need at least a 2-level certificate chain. The UE needs to verify both the signature of gNB's public key and the signature of SIB1, which makes ARIS slower than our proposed scheme. 
The verification phase is performed by UEs, which are usually resource-constrained devices, so keeping a low verification overhead is critical for saving energy, thus extending the battery life. Moreover, as 5G cellular networks are expected to deploy small base stations with much smaller coverage areas, UEs would be forced to switch between base stations at a much faster rate than with conventional base stations. 
The presence of (large numbers of) base stations with small coverage areas would require the UEs to execute the verification phase for the SIB1 message of base stations at a much higher rate, so keeping the verification overhead low is critical for 5G cellular networks. 

\noindent \textit{\underline{End-to-end cryptographic delay:}} We calculate the total cryptographic overhead of the signature schemes instantiated for cellular base station authentication, consisting of the signing and verification cost. The end-to-end delay for certificate-based schemes also includes the verification of the sender's public key authenticity via certificates provided by a certification authority. To be favorable to these schemes, we only consider a certificate chain of size 1. Our scheme provides the lowest end-to-end delay with only 21.48 $\mu$s, making it a very good candidate for authentication, especially in the presence of base stations with small coverage areas. 

\noindent \textit{\underline{Communication overhead:}} SCRA-BGLS has the smallest signature of 29 bytes 
, whereas our scheme \scheme{} and ECDSA-224 have a signature of size 64 bytes. Our scheme has the smallest public key size of 32 bytes, equivalent to ECDSA-256, SchnorrQ, and \texttt{Schnorr-HIBS}. Radio wireless channels are a limited resource, making any cryptographic scheme that adds a lot of communication overhead unlikely to be adopted by network operators. Our scheme provides the smallest communication overhead for authentication. It requires attaching the \scheme{} signature (64 bytes), base station's ID (9 bytes), base station's public key (32 bytes) along with the signing timestamp (4 bytes), and the signature validity period $\mathsf{\bigtriangleup t}$ (2 bytes) for relay attack prevention. This is a total overhead of 111 bytes, which is much lower than the communication overhead for ECDSA-based PKI (277 bytes) and SCRA-BGLS-based PKI (220 bytes). Our scheme can fit in the spare space of SIB1 perfectly. 

\noindent \textbf{Qualitative comparison.}
We compare the schemes based on the type of authentication system and the type of scheme.

\noindent \textit{\underline{System:}} A multi-layer construction is crucial for an identity-based signature scheme suitable for the cellular network architecture (see discussion in Section~\ref{design_decisions}). Both our scheme and \texttt{Schnorr-HIBS} are identity-based and pairing-based. \texttt{Schnorr-HIBS} supports a hierarchal construction while \scheme{} has a fixed two-layer design. On the other hand, signature-based schemes use certificate chains to support multiple signer levels. 

\noindent \textit{\underline{Type of scheme:}} ECDSA, BLS, SCRA-BGLS, SchnorrQ, and ARIS are certificate-based schemes and require a costly public key infrastructure (PKI). On the other hand, \texttt{Schnorr-HIBS} and \scheme{} are identity-based schemes and are more lightweight than certificate-based solutions as they do not require sending huge certificates. However, since the PKI is not available, the user cannot detect a compromised PKG. 

\subsection{Evaluation on 5G Platform}
\noindent \textbf{Implementation Details:}
To measure the performance impact of our scheme, we integrated our scheme into OpenAirInterface. Since the SIB1 message is periodic, we created a separate thread in the base station to pre-generate the next signature. When the current SIB1 message is sent, we calculate the next timestamp and signature for the next SIB message. We also modified the UE part to verify the signature. If the signature is not valid, the UE should abort the registration process. 

\noindent \textbf{Results and Analyzes:}
We measured the SIB1 generation time on the base station, the SIB1 processing time on the UE, and the time difference between the signed timestamp and the time UE verifies the signature. The results are shown in Table \ref{tab:OAI_eval}. During the experiment, the local time of gNB and UE must be synchronized to minimize the errors. To make our result accurate, we run two \textit{USRP B210}s on the same desktop machine. The result shows that our scheme has negligible performance impacts on SIB1 generation and the time difference. The result shows our scheme can generate the SIB with signature in a short time, without adding a heavy load of the base stations. We also observed an 84\% increase in the SIB1 processing time on the UE. However, since the UE does not need to verify the base station very often, the ~5ms time increase can hardly be perceived by the user. 

\begin{table}[t]
    \centering
    \begin{tabular}{|c|c|c|c|}
    \hline
                    &  SIB1 Gen & SIB1 Proc & Time $\Delta$\\
    \hline
    w/o \scheme{}   & 23.99 & 6178.12 & 3.28 \\
    \hline
    w/ \scheme{}    & 24.72 & 11337.08 & 3.40 \\
    \hline
    \end{tabular}
    \caption{Execution time ($\mu$s) in 5G platform}
    \label{tab:OAI_eval}
\end{table}

\section{Related Work}
\label{Related Work}
To address fake base station attacks, some defense techniques rely on improving fake base station detection and then blacklisting them. One such solution is to use measurement reports sent by UEs to detect inconsistencies between the tampered information being broadcast by the fake base stations and the legitimate base station deployment information like the base station identifier or operation frequencies of the base stations~\cite{3GPP:33.501}. Other techniques rely on machine-learning solutions~\cite{jin2019rogue, van2015detecting, engelstad2016strengthening} or  gathering surrounding network signal statistics from the UEs, legitimate base stations or other newly deployed hardware~\cite{dabrowski2014imsi, steig2016network, alrashede2019imsi, dabrowski2016messenger, li2017fbs}. Such techniques can be easily bypassed and have been shown to enable attacks resulting in degradation of network performance and blacklisting of legitimate base stations~\cite{shaik2018impact}. 

Some other approaches \cite{fan2019rehand, alnashwan2023privacy, 10.1145/3658644.3690331} focus on 5G handover to protect user privacy during the change of base stations. Although they also provide base station authentication in the scheme, due to the extra security properties, the performance can be much worse than our signature scheme. Moreover, these approaches require a significant change in the 5G protocol. That makes the adaptation more difficult.

Other IBS schemes \cite{meshram2021iboost, zhang2023identity, au2013realizing} focused on different applications with different security goals. Due to the specific limitations of the 5G system (e.g., message size limit and performance requirements), it is difficult to apply these schemes.

\section{Conclusion}
\label{Conclusion and Future Work}
We have proposed an efficient authentication protocol for 5G networks based on an identity-based signature scheme. Our protocol achieves at least 11 times speedup over the 3GPP proposals and around 2 times speedup over the previous scheme in terms of end-to-end cryptographic delay for authenticating base stations.
Our protocol also achieves the smallest communication overhead against all compared schemes. In addition, We open-source our implementation \cite{E2IBS-github} to benefit future research. 
We hope 3GPP can discuss more about the base station authentication and our solution can be integrated into future specification releases.

\section*{Acknowledgements}
\label{Acknowledgements}
This work was supported by the National Science Foundation (NSF) under grants 2145631, 2215017, 2226447, and 2350213, along with the Defense Advanced Research Projects Agency (DARPA) under contract number D22AP00148. Additional support was provided by the NSF Convergence Accelerator Track G: Securely Operating Through 5G Infrastructure Program, funded by the NSF and the Office of the Under Secretary of Defense—Research and Engineering (ITE 2326898). The research of Rouzbeh Behnia was supported by the USF Office of Research (Sarasota-Manatee campus) through the Interdisciplinary Research Grant Program.

\appendices

\ifCLASSOPTIONcaptionsoff
  \newpage
\fi

\bibliographystyle{IEEEtran}
\bibliography{bibtext}

\begin{IEEEbiography}
[{\includegraphics[width=1in,height=1.25in,clip,keepaspectratio]{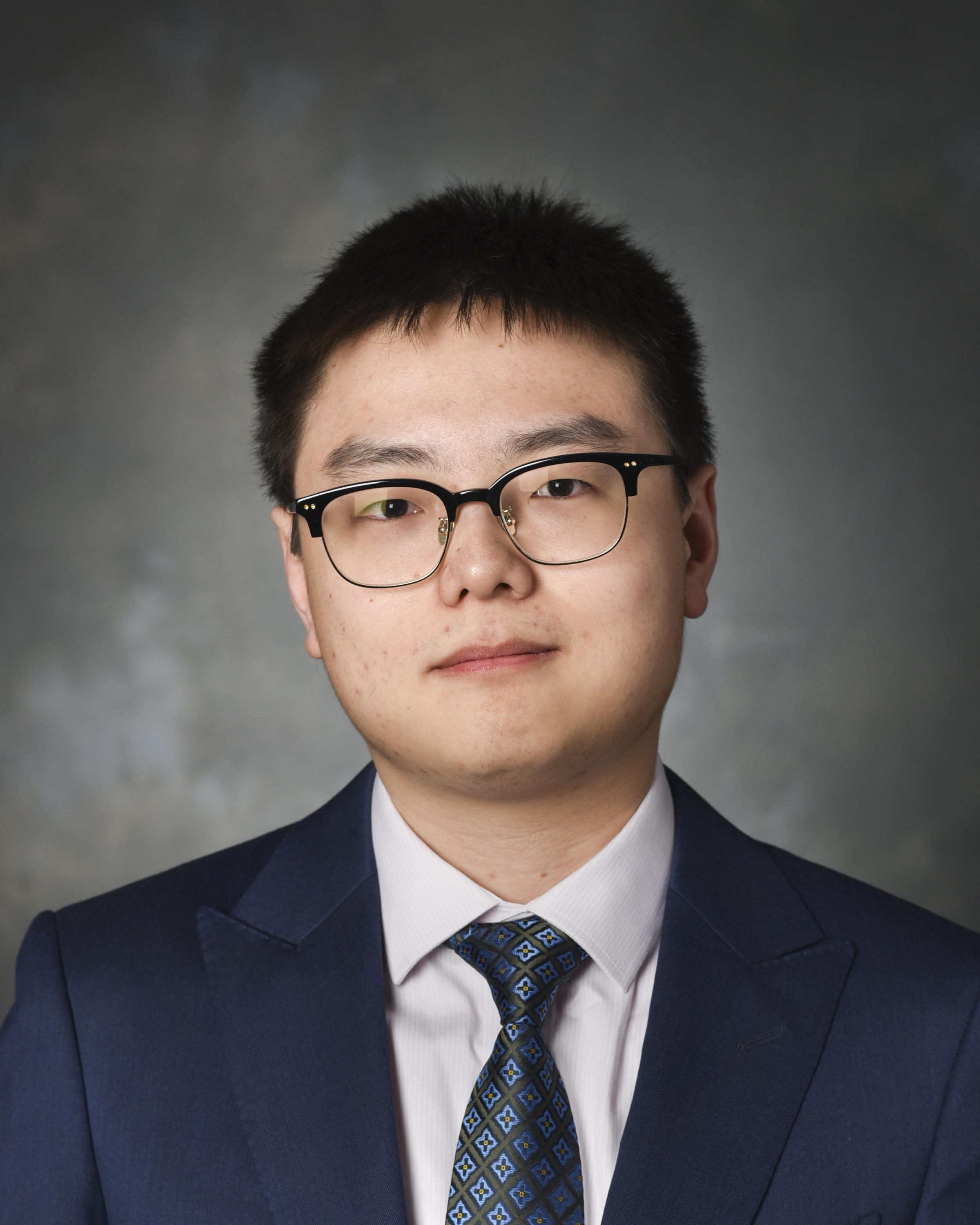}}]
{Yilu Dong}
is a Ph.D. student at Penn State University. He also 
received his M.S. and B.S. degrees from Penn State University. He is interested in communication protocols, software testing, and applied cryptography. His work focuses on the security of 5G systems. Specifically, improving the security and privacy of both UE and core network implementations. 
\end{IEEEbiography}

\begin{IEEEbiography} 
[{\includegraphics[width=1in,height=1.25in,clip,keepaspectratio]{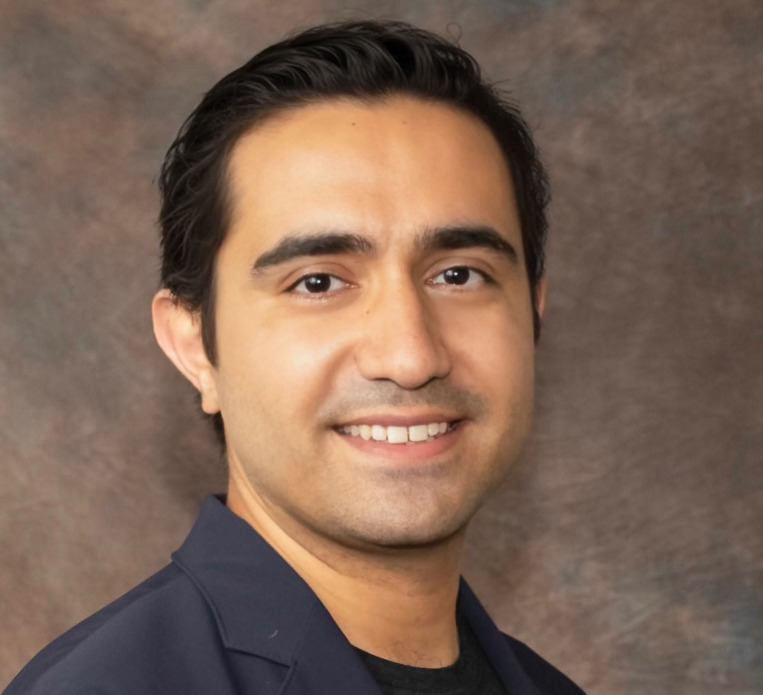}}]
{Rouzbeh Behnia}
is an assistant professor at the School of Information Systems and Management (SISM) at the University of South Florida. He received his Ph.D. in Computer Science from the University of South Florida.
His research focuses on different aspects of cybersecurity and applied cryptography. He is particularly interested in addressing privacy challenges in AI systems, developing post-quantum cryptographic solutions, and enhancing authentication protocols to ensure computation and communication integrity.
\end{IEEEbiography}

\begin{IEEEbiography}[{\includegraphics[width=1in,height=1.25in,clip,keepaspectratio]{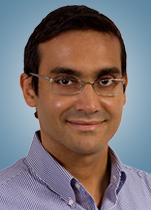}}]{Attila A. Yavuz}
is an Associate Professor in the Department of Computer Science and Engineering and the Director of the Applied Cryptography Research Laboratory at the University of South Florida (USF). He was an Assistant Professor in the School of Electrical Engineering and Computer Science, at Oregon State University (2014-2018) and in the Department of Computer Science and Engineering, USF (2018-June 2021). He was a member of the security and privacy research group at the Robert Bosch Research and Technology Center North America (2011-2014). He received his Ph.D. degree in Computer Science from North Carolina State University (2011). He received his MS degree in Computer Science from Bogazici University (2006) in Istanbul, Turkey. He is broadly interested in the design, analysis, and application of cryptographic tools and protocols to enhance the security of computer systems. Attila Altay Yavuz is a recipient of the NSF CAREER Award, Cisco Research Award (thrice - 2019,2020,2022), and unrestricted research gifts from Robert Bosch (five times). He has authored over 100 products including research articles in top conferences, journals, and patents, some of which world-wide impact with actual deployments.
\end{IEEEbiography}

\begin{IEEEbiography}
[{\includegraphics[width=1in,height=1.25in,clip,keepaspectratio]{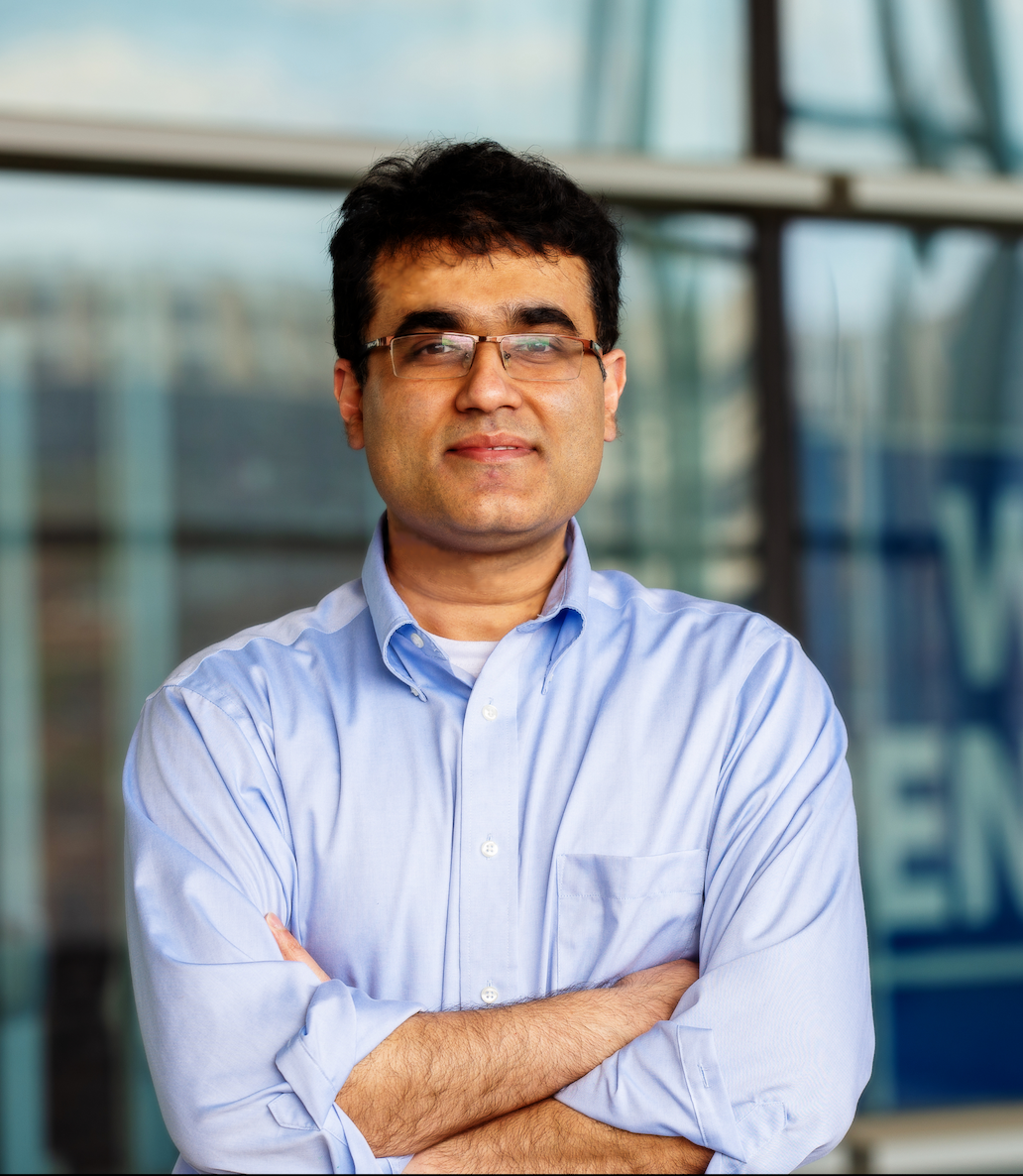}}]
{Syed Rafiul Hussain} (Member, IEEE) is an assistant professor in the Department of Computer Science and Engineering, Pennsylvania State University, State College, PA 16802 USA. His research interests include systems and network security, formal methods, program analysis and applied cryptography. He received a Ph.D. in computer science from Purdue University. He is a Member of IEEE and the Association for Computing Machinery.
\end{IEEEbiography}

\end{document}